%% file: invited_talk_2column.tex
\documentclass[10pt,conference]{IEEEtran}
\input{input.tex}

\usepackage{graphicx}
\usepackage{amssymb}
\usepackage{amsfonts}
\usepackage{amsmath}
\usepackage{latexsym}
\usepackage{epstopdf}
\usepackage{epsfig}

\begin{document}

\newtheorem{thm}{Theorem}
\newtheorem{lemma}{Lemma}
\newtheorem{rem}{Remark}
\newtheorem{exm}{Example}
\newtheorem{prop}{Proposition}
\newtheorem{defn}{Definition}
\newtheorem{cor}{Corollary}
\def\proof{\noindent\hspace{0em}{\itshape Proof: }}
\def\endproof{\hspace*{\fill}~\QED\par\endtrivlist\unskip}
\def\bh{{\mathbf{h}}}
\def\SIR{{\mathsf{SIR}}}
\def\SINR{{\mathsf{SINR}}}

%\def\SINR{{\mathsf{SINR}}}
%\textwidth=6.515in
%%\setlength{\hoffset}{-3mm}%
%\textheight=23.3cm

%\renewcommand{\baselinestretch}{1.4}
%%\textheight=23.3cm

%%\newcommand{\tPout}{\tilde{P}_{\mathsf{out}}}
%%\setlength{\abovecaptionskip}{-5pt}
%\addtolength{\textfloatsep}{-20pt}
%%\setlength{\topskip}{-50pt} \setlength{\parskip}{0pt}
%%\addtolength{\belowcaptionskip}{-5mm}
%\include{header}
%%Tasks:
%% 1: Revised TC definition to include rate per link;
%% 2. Theorem 3: do two results match at the boundary
%% 3: Which limit is zero in Theorem 3?

%\def\Ip{\acute{I}_{\mathsf{P}}}
%\def\Ipp{I_{\mathsf{P}}}
%\def\ddp{\delta_{\mathsf{P}}}
%\def\sp{_{\mathsf{P}}}
%\def\ss{_{\mathsf{S}}}
%\def\fa{\frac{2}{\alpha}}
\title{Average Throughput and Approximate Capacity of Transmitter-Receiver Energy Harvesting Channel with Fading}
\author{Jainam Doshi and Rahul Vaze}

\maketitle
\begin{abstract} We first consider an energy harvesting channel with fading, where only the transmitter harvests energy from natural sources. We bound the optimal average throughput by a constant for a class of energy arrival distributions. The proposed method also gives a constant approximation to the capacity of the energy harvesting channel with fading. Next, we consider a more general system where both the transmitter and the receiver employ energy harvesting to power themselves. In this case, we show that finding an approximation to the optimal average throughput is far more difficult, and identify a special case of unit battery capacity at both the transmitter and the receiver for which we obtain a universal bound on the ratio of the upper and lower bound on the average throughput.
\end{abstract}

\section{Introduction}
Finding optimal power/energy transmission policies to maximize the long-term throughput in an energy harvesting (EH) communication system is a challenging problem and has remained open in full generality. Structural results are known for the optimal solution \cite{sinha2012optimal}, however, explicit solutions are only known for a sub-class of problems, for example, binary transmission power 
\cite{michelusi2012optimal}, discrete transmission power \cite{VazeEHICASSP14}, etc. Recently, some progress has been reported in approximating the per-slot throughput by a universal constant in \cite{dong2014near}, for an AWGN channel. 

In this paper, we approximate the per-slot throughput of the EH system with fading by a universal constant for a class of energy arrival distributions. The fading channel problem is more challenging than the AWGN case, since the energy/power transmitted per-slot depends on the realization of the channel unlike the AWGN problem. Thus, finding an upper bound on the throughput is hard. We take recourse in Cauchy-Schwarz inequality for this purpose, and then surprisingly using a channel independent power transmission policy proposed in \cite{dong2014near}, show that the upper and lower bound on the per-slot throughput (or expected throughput) differ at most by a constant. Using the techniques of \cite{dong2014near}, we also show that our universal bound also provides an approximation of the Shannon capacity of the energy harvesting channel with fading upto a constant.

In addition to EH being employed at the transmitter, a more relevant or practical scenario is when EH is employed at both the transmitter and the receiver. The EH setting at the receiver is simpler than at the transmitter, since the only decision the receiver has to make is whether to stay {\it on} or not. In the {\it on} state, the receiver consumes a fixed amount of energy, so the energy consumption model at the receiver is Bernoulli. 

With EH at both the transmitter and the receiver, there is inherent lack of information about the receiver energy levels at the transmitter and vice-versa. One can show that the optimal policy at both the transmitter and the receiver is of threshold type, but the thresholds depend on both the energy states in a non-trivial way. Because of the common fading channel state that is revealed to both the transmitter and receiver for each slot, both the transmitter and the receiver have some partial statistical information about others' energy state which is important for finding the optimal policy. 

We show that, in general, it is difficult to bound the gap between the upper and lower bound when EH is employed at both the transmitter and the receiver. Then we identify a special case of unit battery capacity at both the transmitter and receiver, and where the transmitter operates with binary transmission power, for which we propose a strategy that achieves at least half of the upper bound on the per-slot throughput, giving a ratio bound.

\section{System Model}

We consider slotted time, and a single transmitter-receiver pair, where the transmitter harvests energy from the environment. 
Let $E_t$ be the amount of energy harvested at time step $t$ which is stored in a battery of size $B_{max}$. 
The energy harvested at each time step $E_t$ is a discrete time ergodic and stationary random process. At time $t$, the fading channel between the transmitter and the receiver is $h_t$, where $h_t$ is assumed to be i.i.d. Rayleigh distributed for each $t$. Thus, $|h_t|^2 \sim \exp(1)$. Let $B_t$ be the battery energy level at time $t$. Let energy used at time $t$ given channel $h_t$ be $P_h(t) \le B_t$, then the rate obtained at time $t$ is 
$r(t) = \frac{1}{2} \log (1 + h_tP_h(t))$. 
The energy state at the transmitter evolves as $B_t = B_{t-1} + E_{t-1} - P_h(t-1)\b1_{P_h(t-1)\ge B_{t-1}}$.

The long-term throughput is defined as 
\begin{equation}\label{eq:ltt}
T = \lim_{N\rightarrow \infty}\frac{1}{N} \sum_{t=1}^N r(t).
\end{equation}
Our objective is to find the optimal energy consumption $P_h(t)$ at the transmitter, given the energy neutrality constraint $P_h(t) \le B_t$, that maximizes $T$, i.e. $T^{\star} = \max_{P_h(t)\le B_t} T$.  

\section{An Upper Bound on maximum achievable throughput}

In this section, we upper bound the maximum long term throughput $T^{\star}$. 

\begin{defn}
A random variable $X$ defined over $(\Omega, {\cal F}, {\cal P})$ is called an increasing random variable if for $\omega _1,\omega _2 \in \Omega, X(\omega _1) \leq X(\omega _2)$ whenever $\omega _1 \leq \omega _2$.  

\end{defn}

\begin{lemma}
$\left( \textbf{FKG Inequality} \right)$ \cite{Grimmett1980} For increasing random variables $X$ and $Y$, $\mathbb{E}\left[ XY \right] \geq \mathbb{E}[X] \mathbb{E}[Y].$

\end{lemma}

\begin{thm}\label{thm:ub}For any energy consumption policy $P_h(t)$ such that $P_h(t) \le B_t$,
\begin{equation} 
T^{\star} \leq \frac{1}{2} \log \left(1 + \sqrt{2} \sqrt{\mathbb{E}[E_t^2]} \right).
\end{equation}

\end{thm}

\begin{proof}By ergodicity, \eqref{eq:ltt} is equal to 
\begin{equation}
\label{eq:ergodicity}
T = \bbE \left\{\frac{1}{2} \log(1 + h_tP_h(t)) \right\}.
\end{equation}

By Jensen's inequality, for any  $P_h(t)$,
\begin{equation}
\label{eq:jensen}
\mathbb{E} \left[ \log \left( 1 + h_tP_h(t) \right) \right] \leq \log \left( 1 + \mathbb{E}\left[ h_tP_h(t) \right] \right).
\end{equation}
Applying Cauchy-Schwarz inequality for any two random variables $X, Y$, $\bbE\{XY\}^2 \le \bbE\{X^2\}\bbE\{Y^2\}$, we have
\begin{equation}
\label{eq:CSineq}
\left(\mathbb{E}[h_tP_h(t)] \right)^2 \leq \mathbb{E}[h_t^2] ~ \mathbb{E}[P_h^2(t)].
\end{equation}

By energy neutrality constraint, we have,
\begin{align}
\nonumber
\left( \frac{1}{N} \sum\limits_{t=1}^{N}P_{h_t}(t) \right)^2 &\leq \left( \frac{1}{N} \sum\limits_{t=1}^{N}E_t \right)^2 ~\forall N,\\
\nonumber
\frac{1}{N^2-N} \left( \sum\limits_{t_1 \ne t_2} P_h(t_1) P_h(t_2) \right) & \stackrel{(a)}{\leq} \left( \frac{1}{N} \sum\limits_{t=1}^{N}E_t \right)^2 ~\forall N,\\
\label{eq:normal}
\mathbb{E} \left[ P_h(t_1) P_h(t_2) \right] &\stackrel{(b)}{\leq} \mathbb{E}\left[ E(t)^2 \right],
\end{align}
where $(a)$ is obtained by dropping terms of the type $P_h(t_i)^2$, and $(b)$ is obtained by taking the limit as $N \rightarrow \infty$, and where the R.H.S. follows since $E(t)$ is i.i.d. across $t$.

Note that $P_h(t)$ is an increasing random variable in the channel fade state $h$. This is clear as the optimal energy utilization strategy would spend a greater amount of energy for a  higher channel fade state $h$ which implies that $P_{h_1}(t) \leq P_{h_2}(t)$ whenever $h_1 \leq h_2$. Similarly, $P_h(t)$ is an increasing random variable in the battery energy level $B_t$. Therefore, from Lemma 1,

\begin{equation}
\nonumber
\mathbb{E} \left[ P_h(t_1) P_h(t_2) \right] \geq \mathbb{E} [P_h(t_1)] \mathbb{E} [P_h(t_2)].
\end{equation}

As $P_h(t_i)$ are identically distributed, we have
\begin{equation}
\label{eq:FKG1}
\mathbb{E} \left[ P_h(t_1) P_h(t_2) \right] \geq \left( \mathbb{E} [P_h(t)] \right)^2.
\end{equation}

From \eqref{eq:normal} and \eqref{eq:FKG1}, 

\begin{equation}
\label{eq:FKG}
\mathbb{E}\left[P_h^2(t) \right] \leq \mathbb{E}\left[ E_t^2 \right].
\end{equation}

Thus, from \ref{eq:jensen}, \ref{eq:CSineq} and \ref{eq:FKG}, we have
\begin{equation}
\label{eq:ub}
T^{\star} \leq \frac{1}{2} \log \left(1 + \sqrt{\mathbb{E} \left[ h_t^2 \right]} \sqrt{\mathbb{E}[E_t^2]} \right).
\end{equation}
%Thus, 
%\begin{eqnarray}
%\nonumber
%&\lim\limits_{N \rightarrow \infty}~\frac{1}{N}\left( \sum\limits_{t=1}^{N}P_{h_t}(t) \right)^2 \leq \lim\limits_{N \rightarrow \infty}~\frac{1}{N}\left( \sum\limits_{t=1}^{N}E_t \right)^2.
%\end{eqnarray}

As $h_t$ is an exponentially distributed random variable with a mean of unity, we have that $\mathbb{E}[h_t^2] = 2$ which proves the theorem.

\end{proof}
We denote this upper bound on the achievable throughput by $T_{ub}$. We next propose an energy allocation strategy, and compare the average throughput obtained by it with $T_{ub}$.

\section{Achievable Strategy} To build intuition on the proposed achievable strategy, we begin with the simple Bernoulli energy arrival process.
\subsection{Bernoulli Energy Arrival}
Let the energy arrival at time $t$ be i.i.d. Bernoulli with $E_t = E$ with probability $p$, and $E_t = 0$ with probability $1-p$.
First we look at the $E>B_{max}$ case, where it is sufficient to consider $E_t = B_{max}$. For $E_t = B_{max}$, each energy arrival (called epoch) completely fills up the battery, and hence the throughput obtained in each epoch is i.i.d.. Thus, it is sufficient to consider any one epoch to lower bound the expected throughput. 

We use the same strategy (call it Constant Fraction Policy (CFP)) as proposed in \cite{dong2014near} for an AWGN channel, where 
\begin{equation}
P_h(t) = p(1-p)^jB_{max}, ~~ \text{for} ~~ j=0,1,2, . . .~.
\end{equation}
where $j = t - \max \left\lbrace t':E_{t'}=E, \forall t' \leq t \right\rbrace$ is the number of channel uses since the last epoch. Note that between any two epochs $T_1$ and $T_2$, the energy used is $\sum_{t=T_1}^{T_2} = \sum\limits_{j=0}^{\infty}p(1-p)^jB_{max} = B_{max}$, thus CFP satisfies energy neutrality constraint. 

Note that CFP is independent of the channel fade state $h$, and hence seems to be highly sub-optimal. Analysis with any natural $h$ dependent policy, however, is not immediately amenable for analysis. 
Using the concavity of the $\log$ function, we show that even an $h$ independent policy has a bounded gap from the upper bound for a class of energy arrival distributions. For Bernoulli energy arrivals, we cannot bound the gap universally and the gap depends on the rate $p$. 
\begin{lemma}\label{lem:cfpB}
The average throughput obtained by CFP with Bernoulli energy arrival distribution is given by,
\begin{equation}
T_{lb} = \sum\limits_{j=0}^{\infty}p(1-p)^j \int\limits_{0}^{\infty} \frac{1}{2} \log (1+hp(1-p)^jB_{max})e^{-h}dh.
\end{equation} 
\end{lemma}

\begin{proof}
As stated before, we consider a single epoch to lower bound the average throughput of CFP. 
Let the first epoch be $T_1=0$ and $\tau_1$ be the time between the first two epochs. For ease of notation, 
let $P_h(t) = p(1-p)^jB_{max}= P(j)$ for $j=t=0, \dots, \tau_1-1$. Let $T_{lb}$ be the expected throughput of the CFP. Then by renewal reward theorem, 
\begin{align}
\nonumber
T_{lb} &= \frac{\mathbb{E}\left[ \sum\limits_{j=0}^{\tau_1-1}  \frac{1}{2}\log(1 + hP(j) \right]}{\mathbb{E}[\tau_1]},\\
\nonumber
& \stackrel{(a)}{=}p\sum\limits_{i=1}^{\infty} \mathbb{P}(\tau_1=i)\sum\limits_{j=0}^{i-1}\int\limits_{0}^{\infty}\frac{1}{2}\log(1 + hP(j))e^{-h}dh,\\
\nonumber
& \stackrel{(b)}{=}p \sum\limits_{j=0}^{\infty} \left( \sum\limits_{i=j+1}^{\infty}(1-p)^{i-1}p \right) \int\limits_{0}^{\infty}\frac{1}{2}\log(1+hP(j))e^{-h}dh\\
\nonumber
& \stackrel{(c)}{=} \sum\limits_{j=0}^{\infty}p(1-p)^j \int\limits_{0}^{\infty} \frac{1}{2} \log (1+hp(1-p)^jB_{max})e^{-h}dh.
\end{align}
$(a)$ follows because $\tau_1$ is a geometric random variable with parameter $p$ with $\mathbb{E}[\tau_1] = 1/p$, $(b)$ follows since $\mathbb{P}(\tau_1 = i) = (1-p)^{i-1}p$, and thereafter interchanging the order of summations, and $(c)$ follows using $\sum\limits_{i=j+1}^{\infty}(1-p)^{i-1}p = (1-p)^j$, and substituting back $P(j) = p(1-p)^jB_{max}$.

\end{proof}

\begin{comment}
From Theorem 1, we know that the maximum achievable average throughput is upper bounded by,
\begin{equation}
\nonumber
T_{ub} \stackrel{(a)}{=} \frac{1}{2}\log \left(1 + \sqrt{2p}B_{max} \right),\\
\end{equation}
where $(a)$ follows from theorem 1 by substituting $\mathbb{E}[E_t^2] = pB^2_{max}$ for Bernoulli energy arrivals.
\end{comment}

Next, we bound the gap between upper bound $T_{ub}$ (Theorem \ref{thm:ub}) and the $T_{lb}$ of CFP from Lemma \ref{lem:cfpB}.

\begin{lemma}
$T_{ub} - T_{lb} \leq \frac{1}{2} \log \left( 1 + \sqrt{2p}~k \right),$

where $k$ satisfies
\begin{align}
\nonumber
\frac{1}{2} \log \left( 1 + \sqrt{2p}~k \right) =&~ 0.54 - \frac{1}{4} \log (p) + \frac{1}{2\ln 2}\frac{1}{\sqrt{2p}~k}\\
\label{eqn:recursion}
& + \frac{1-p}{2p} \log \left(\frac{1}{1-p} \right).
\end{align}
\end{lemma}

In general, for Bernoulli energy arrivals, the gap between the upper bound and the achievable expected throughput of the CFP is not universally bounded for all values of $p$. This negative result is, however, only limited to Bernoulli energy arrivals. We show a universal bound on the gap between the upper bound and the achievable expected throughput of the CFP with other distributions using the bound on the $T_{ub} - T_{lb}$ for the Bernoulli case with 
$p=0.5$.
\begin{lemma}
For Bernoulli energy arrivals with $p=0.5$, $T_{ub} - T_{lb} \le 1.41 \ \text{bits}$. 
\end{lemma}
\begin{proof}
Fixing $p=0.5$, the solution to \eqref{eqn:recursion} is $k=6.05$. Then, from Lemma 2,
$T_{ub} - T_{lb}\leq \frac{1}{2} \log \left( 1 + \sqrt{2p}~k \right) =
1.41$.

\end{proof}

When the battery size $B_{max} \ge E$, to derive an achievable strategy we assume the battery size $B_{max} = E$ which is less than the actual battery capacity, $B_{max}$, and use CFP.  Clearly, it is a feasible strategy and one can easily obtain the same bounds as for the $B_{max} < E$ case.

%\begin{lemma}
%For Bernoulli energy arrivals with $p=0.5$ and packet size $E<B_{max}$, the approximation gap between $T_{ub}=\frac{1}{2} \log \left( 1 + E \right)$, and the throughput achieved by CFP, $T_{lb}$, is upper bounded by 1.41 bits.
%\begin{equation}
%\frac{1}{2}\log(1 + E) - T_{lb} \leq 1.41, ~\forall~E < B_{max}~\text{and}~p=0.5.
%\end{equation}
%\end{lemma}

\section{Generalization to other Energy Profiles}
Let $X_t$ denote the energy arriving in time $t$, i.e. $X_t = E_t$ with CDF $F_X(\cdot)$. 
Let $\delta$ be such that $F_X(\delta) = 0.5$. We will assume that $\delta \leq B_{max}$. The case $\delta > B_{max}$ needs some modifications but can be worked out similarly. Thus, we know that with probability $p=0.5$, $X_t > \delta$. 
We now propose to use CFP as if the energy arrival process were i.i.d. Bernoulli with fixed size $\delta$ and $p=0.5$. Thus, the actual energy stored in the battery is $0$ if $X_t \le \delta$, and $\delta$ if $X_t < \delta$ 

%By definition, the probability $p$ that the energy arrival at time $t$, $E_t$, is greater than $\delta$ is $0.5$.  The throughput achieved is denoted as before by $T_{lb}$. By this, we mean that we assume that there is no energy arrival when the arrival energy packet size is smaller than $\delta$, and we assume that the energy packet size is equal to $\delta$ each time we receive an energy packet of size atleast $\delta$. Clearly, this is an energy feasible strategy for general i.i.d. energy arrival process but may seem highly suboptimal and wasteful of energy. However, we have the following theorem for the throughput achieved $T_{lb}$.

\begin{thm}
The expected throughput achieved by CFP $T_{lb}$ satisfies $T_{lb} \geq T_{ub} - 1.67 - \frac{1}{4} \log \left( \frac{\mathbb{E}[X^2]}{\left( F_X^{-1}(0.5) \right)^2} \right).$
\end{thm}

\begin{cor}When $E_t$ is uniformly distributed between 0 and $B_{max}$, CFP achieves
\begin{equation}
T_{lb} \geq T_{ub} - 1.76.
\end{equation}
%where $T_{ub} = \frac{1}{2} \log \left( 1 + \sqrt{2}\frac{B_{max}}{\sqrt{3}} \right)$ for uniform energy arrivals from Theorem 1.
\end{cor}
\begin{proof}
The corollary follows by substituting $F_X^{-1}(0.5) = \frac{B_{max}}{2}$ and $\mathbb{E}[X^2] = \frac{B_{max}^2}{3}$ for uniform energy arrivals between 0 and $B_{max}$ in Theorem 2.

\end{proof}
Using this universal bound, we can get bounds on the capacity of this channel similar to Theorem 9 of \cite{dong2014near}.
\begin{thm}
The capacity $C$ of the fading channel with EH is bounded by
\begin{equation}
T_{lb}  - \frac{1}{4} \log \left( \frac{\mathbb{E}[X^2]}{\left( F_X^{-1}(0.5) \right)^2} \right) -c \le C \le T_{ub},
\end{equation} where $c$ is a constant that depends on the distribution of $X$.
\end{thm}

%\begin{rem}
%The proposed strategy provides a bounded approximation gap to the upper bound w achievable throughput $T_{ub}$, for energy profiles with finite second moment. However, it fails to provide a bounded gap for energy profiles with unbounded second moment although such profiles may not be common in practice. 
%\end{rem}

\begin{comment}

Thus, we can bound the gap between the maximum achievable throughput and the throughput achieved by our proposed strategy for most of the energy arrival profiles with the gap being a function of  $\frac{\mathbb{E}[X^2]}{\left( \mathbb{E}[X] \right)^2}$. However, it is possible to come up with counter examples for which the proposed method will fail to give a bounded gap. However, such profiles may not be common in practice.

\end{comment}

\section{Receiver Energy Harvesting}
In this section, we assume that receiver also uses energy harvesting to power itself. Compared to the transmitter, however, the receiver structure/decision is simpler; it only has to decide whether to stay {\it on} or {\it off} in any given slot. When the receiver is {\it on}, it consumes a fixed amount of energy. The receiver is assumed to have a finite battery of size $\tilde{B}_{max}$, and energy arrival at time $t$ is $\tilde E_t$. Note that the transmitter and receiver are separated and do not access to each others energy availability information. At each time $t$, only the channel $h_t$ is revealed to both of them, using which the transmitter has to decide how much power to transmit, and the receiver has to decide whether to stay {\it on} or not.  Let $\b1_R(h_t,t) =1$ if receiver is {\it on} at time $t$, otherwise $0$. Then the rate obtained at time $t$ is 
${\tilde r}(t) = \textbf{1}_R(h_t,t) \log \left( 1 + h_tP_h(t) \right)$, and long-term throughput is 
\begin{equation}\label{eq:rltt}
{\tilde T} = \lim_{N\rightarrow \infty}\frac{1}{N} \sum_{t=1}^N {\tilde r}(t).
\end{equation}
Our objective is to find optimal $P_h(t)$ and $\textbf{1}_R(h_t,t)$, given the energy neutrality constraint $P_h(t) \le B_t$, that maximizes ${\tilde T}$, i.e. ${\tilde T}^{\star} = \max_{P_h(t)\le B_t, \textbf{1}_R(h_t,t)} {\tilde T}$.

We next show that with receiver energy harvesting, the problem of approximating ${\tilde T}^{\star}$ is much harder, and finding an universal bound on the difference of the upper and the lower bound is challenging, since the transmitter-receiver decisions about $P_h(t)$ and $\textbf{1}_R(h_t,t)$ are intimately connected via the revealed channel $h_t$. 

If we use the Cauchy-Schwarz inequality twice, first on $\textbf{1}_R(h_t,t), \log \left( 1 + h_tP_h(t) \right)$, and next on $h_t, P_h(t)$, we get   
\begin{equation}
\label{eq:upper bound}
{\tilde T}^{\star} \leq 2~\sqrt{\mathbb{E} \left[ \left( \tilde{E}_t \right)^2 \right]} \log \left( 1 + \sqrt{2} \sqrt{\mathbb{E}\left[ (E_t)^2 \right]} \right).
\end{equation}
In the simplest case, when the energy harvesting distribution at the receiver is Bernoulli with rate $q$, we get 
\begin{equation}
\label{eqn:max troughput}
\tilde{T}^{\star} \le \tilde{T}_{ub} = 2~ \sqrt{q} \log \left( 1 + \sqrt{2} \sqrt{\mathbb{E}\left[ (E_t)^2 \right]} \right).
\end{equation}

For an achievable strategy, consider a simple receiver strategy, where $\textbf{1}_R(h_t,t) = 1$, whenever receiver has sufficient energy, and $\textbf{1}_R(h_t,t)=0$ otherwise.
From Lemma \ref{lem:cfpB}, the throughput achieved by using CFP at the transmitter is,
\begin{equation}\nn
\tilde{T}_{lb} =q \sum\limits_{j=0}^{\infty}p(1-p)^j \int\limits_{0}^{\infty} \frac{1}{2} \log (1+hp(1-p)^jB_{max})e^{-h}dh.
\end{equation}  The main difference in $\tilde{T}_{ub}$ and $\tilde{T}_{lb}$ is the $\sqrt{q}$ term and the $q$ term, respectively.  Even if use a channel dependent strategy where $\textbf{1}_R(h_t,t) = 1$ if $h > \gamma$, even then getting pre-log term of $\sqrt{q}$ is not possible. Inherently in this case, the upper bound is too loose because of the coupled transmitter-receiver optimal decisions.
%Note that $\tilde{T}_{ub}=2\sqrt{q}~T_{ub}$ and $\tilde{T}_{lb}=qT_{lb}$ where $T_{ub}$, and $T_{lb}$, are as defined before. From Theorem 2, $T_{ub}-T_{lb}$ is bounded for energy arrivals with finite second moment at the transmitter. However, for high energy arrival rate at the transmitter, $\frac{\tilde{T}_{ub}}{\tilde{T}_{lb}} = \frac{2}{\sqrt{q}}\frac{T_{ub}}{T_{lb}} \approx \frac{2}{\sqrt{q}}$ which goes unbounded for $q \rightarrow 0$.  
We next present a special case for which we can bound the ratio of the upper and the lower bound with receiver energy harvesting. 
\subsection{$B_{max} = \tilde{B}_{max}=1$}In this section, we consider a special case of $B_{max} = \tilde{B}_{max}=1$, and binary energy transmission policy at transmitter, i.e. $P_h(t) \in \{0,1\}$. Also, the energy arrivals at the transmitter and the receiver are assumed to be Bernoulli with parameter $p$ and $q$ respectively. Thus, $
{\tilde r}(t) = \textbf{1}_{T}(h_t,t) \textbf{1}_{R}(h_t,t) \log (1 + h_t)$, and we want to maximize \eqref{eq:rltt}, with respect to $\textbf{1}_{T}(h_t,t) \textbf{1}_{R}(h_t,t)$ under the energy neutrality constraint at both the transmitter and the receiver.
%
%For this sytem, the throughput obtained when the channel fade state is $h_t$, is given by,
%\begin{equation}
%r(t) = \textbf{1}_{T}(h_t,t) \textbf{1}_{R}(h_t,t) \log (1 + h_t),
%\end{equation}
%where $\textbf{1}_{T}(h_t,t)$, and $\textbf{1}_{R}(h_t,t)$, are the indicator random variables corresponding to the energy usage at transmitter and receiver respectively. Note that $\textbf{1}_T(h_t,t) = 1$ if the transmitter spends one unit of energy and $\textbf{1}_T(h_t,t) = 0$ otherwise.
%
%The maximum achievable throughput, $\tilde{T}_{max}$, is then given by,
%\begin{equation}
%\tilde{T}_{max} = \max\limits_{\textbf{1}_{T}(h_t,t),\textbf{1}_{R}(h_t,t)}~\mathbb{E}\left[ \textbf{1}_{T}(h_t,t) \textbf{1}_{R}(h_t,t) \log (1 + h_t) \right].
%\end{equation}

\begin{lemma}
$\tilde{T}^{\star} \leq \tilde{T}_{ub} = \min \left\lbrace p,q \right\rbrace \int\limits_{\gamma^{*}}^{\infty} \log \left( 1 + h \right) e^{-h}dh$,
where $\gamma^{*} = - \ln \left( \min \left\lbrace p,q \right\rbrace \right)$.
\end{lemma}

\begin{proof}
Let us assume that $p > q$. To upper bound $\tilde{T}^{\star}$ assume that the transmitter always has energy to transmit i.e. $\textbf{1}_{T}(h_t,t) = 1, \ \forall \ t$. Thus, we have a system with energy harvesting only at the receiver, for which the optimal transmission policy is known be of threshold type \cite{michelusi2012optimal}, \cite{sinha2012optimal} i.e. for $\tilde{B}_t=1$, $\textbf{1}_R(h_t,t) = 1$ if $h_t > \gamma$, and $0$ otherwise.
Next, we argue that the optimal threshold $\gamma^{*}$ will satisfy $\mathbb{P}(h > \gamma^{*}) = q$. 
\begin{itemize}
\item If  $\gamma > \gamma^{*}$, the energy arrival rate at the receiver is greater than the energy usage at the receiver, and essentially there is a wastage of energy  resulting in sub-optimality.
\item If $\gamma < \gamma^{*}$, the energy usage rate is faster than energy arrival rate $q$. Thus, the receiver will remain {\it on} for weak channel states, and will not have sufficient energy to transmit on the better channel gains. 
\end{itemize}
%Note that $\mathbb{P}(h>\gamma^{*})=q$ implies $\gamma^{*} = - \ln q$ as $h$ is exponentially distributed with mean 1.

Thus, the maximum achievable throughput is
\begin{equation}
\tilde{T}^{\star} \le \tilde{T}_{ub} = q \int\limits_{\gamma^{*}}^{\infty} \log (1 + h) f(h)dh.
\end{equation}
 The case of $q>p$, can be proved by interchanging the role of the transmitter and the receiver.
  
\end{proof}

Next, we propose a \textit{Common Threshold Policy} (CTP) to lower bound the achievable rate.
i) Transmitter Policy : The transmitter simulates an i.i.d. Bernoulli random variable $X_1$ for each slot with parameter $q$. 
The transmitter waits till $X_1=1$, and thereafter transmits whenever $h_t > \gamma^{*}$ if $B_t=1$.
ii) Receiver Policy : The receiver remains {\it on} whenever $h_t > \gamma^{*}$ if $\tilde{B}_t=1$.

\begin{thm}
The throughput achieved by CTP, $\tilde{T}_{lb}$, satisfies,
\begin{equation}
\tilde{T}_{lb} \geq \frac{1}{2}~\tilde{T}_{ub}.
\end{equation}
\end{thm}

\begin{proof}
Consider the case of $p>q$. Note that the receiver is following the optimal strategy with CTP. We need to show that with probability $\frac{1}{2}$, the transmitter is ${\it on}$ whenever receiver is. The random variable $X_2 = \left\{\textbf{1}_{h_t>\gamma^{*}}\right\}$ is  Bernoulli with parameter $q$ as $\mathbb{P}\left(h_t>\gamma^{*} \right) = q$.  We call this random variable $X_2$.
So whenever $X_1 < X_2$, the transmitter and receiver are {\it on} for slots where $\left\{\textbf{1}_{h_t>\gamma^{*}}\right\}$. Since $X_1$ and $X_2$ are identical random variables, $P(X_1<X_2) = \frac{1}{2}$. Thus, CTP can achieve at least half of the throughput of the upper bound $\tilde{T}_{ub}$.

\end{proof}

\section{Simulations}
\begin{figure}[h]
\centering
\includegraphics[scale=0.5]{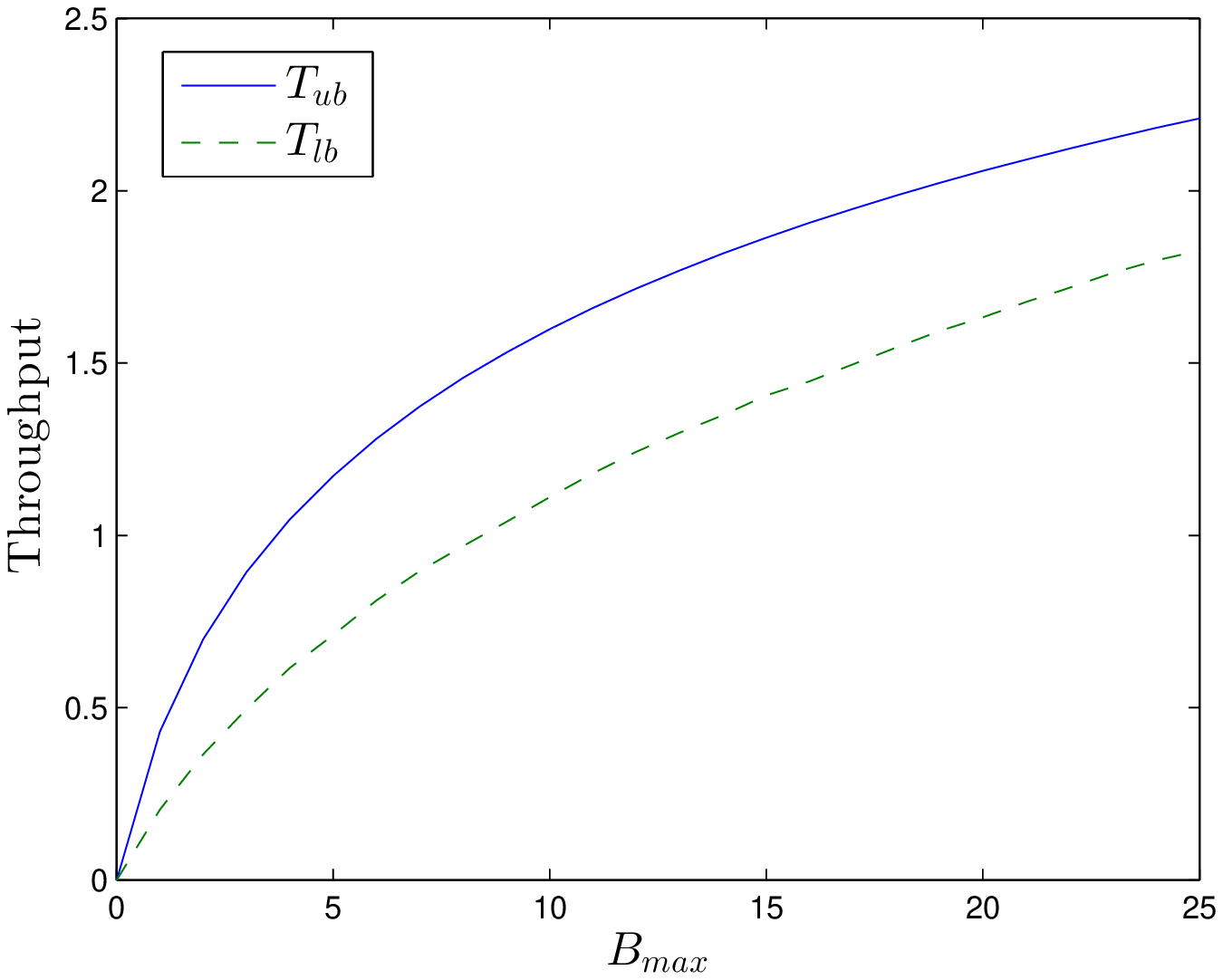}
%\vspace{-0.2in}
\caption{Performance of the single node optimal policy for $B=10$. The solid curve represents optimal payoff under uniform energy arrivals in [0:10], and the dashed curve corresponds to binary energy arrivals and transmission.}
\label{fig:singlesource}
\end{figure}

\begin{figure}[h]
\centering
\includegraphics[scale=0.5]{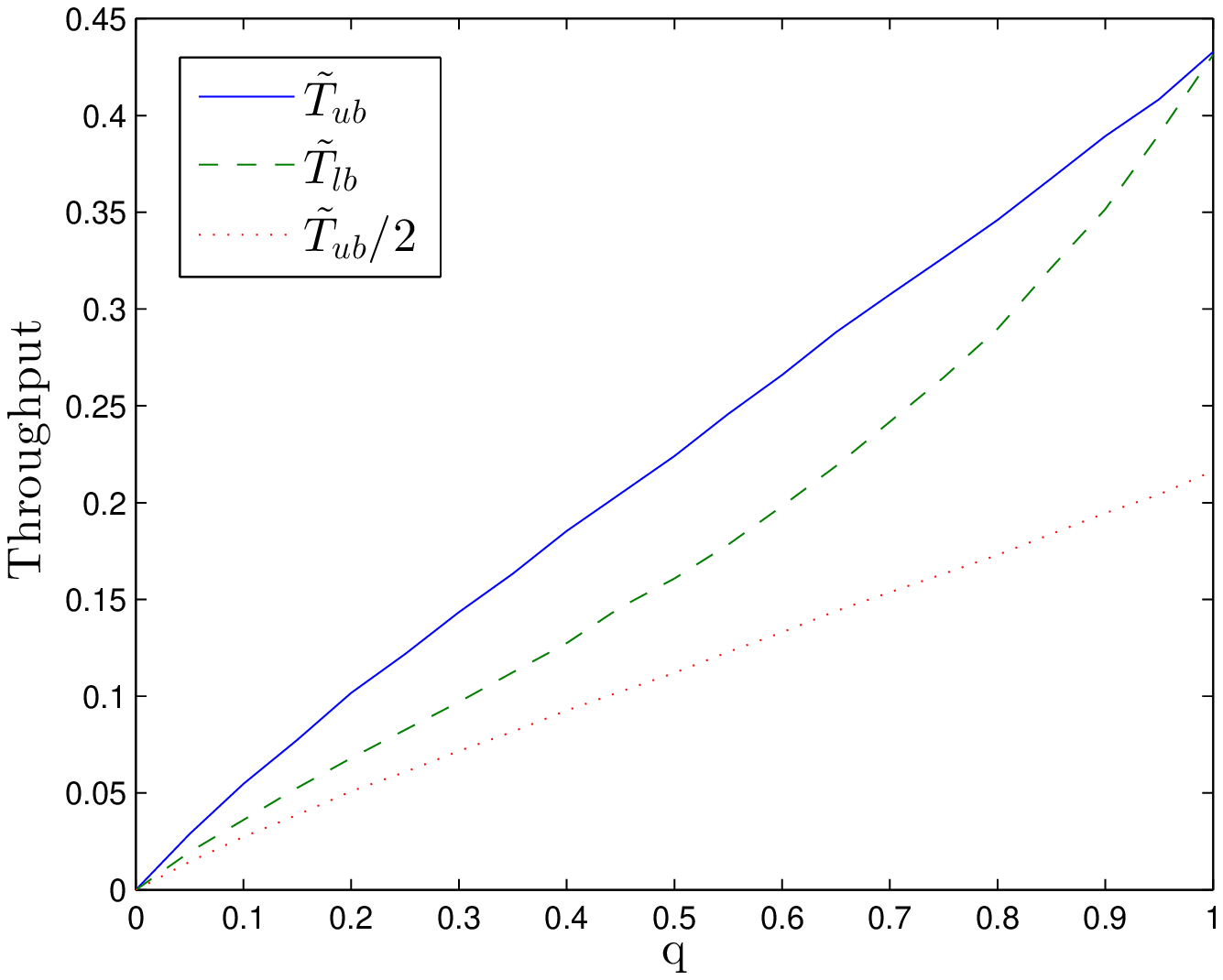}
%\vspace{-0.2in}
\caption{Performance of the single node optimal policy for $B=10$. The solid curve represents optimal payoff under uniform energy arrivals in [0:10], and the dashed curve corresponds to binary energy arrivals and transmission.}
\label{fig:singlesource}
\end{figure}

\appendices
\section{}
\begin{proof}
For $B_{max}<k$, $T_{ub} - T_{lb}$
\begin{align}
\nonumber&=\frac{1}{2}\log \left(1 + \sqrt{2p}B_{max} \right) \\
\nonumber
&-\sum\limits_{j=0}^{\infty}p(1-p)^j \int\limits_{0}^{\infty} \frac{1}{2}\log(1+hp(1-p)^jB_{max}) e^{-h}dh,\\
\nonumber
&\stackrel{(a)}{\leq} \frac{1}{2}\log \left(1 + \sqrt{2p}~k \right),
\end{align}
where $(a)$ follows because $\sum\limits_{j=0}^{\infty}p(1-p)^j \int\limits_{0}^{\infty} \frac{1}{2}\log(1+hp(1-p)^jB_{max})e^{-h}dh \geq 0$ as the integrand is always positive, and $B_{max} < k$.

For $B_{max} > k$, $T_{ub}-T_{lb}$
\begin{align}
\nonumber
=&\frac{1}{2}\log \left(1 + \sqrt{2p}B_{max} \right) - \\
\nonumber
&\sum\limits_{j=0}^{\infty}p(1-p)^j \int\limits_{0}^{\infty} \frac{1}{2}\log(1+hp(1-p)^jB_{max})e^{-h}dh,\\
\nonumber
\stackrel{(a)}{\leq}&\frac{1}{2}\log \left(1 + \sqrt{2p}B_{max} \right) \\
\nonumber
&- \sum\limits_{j=0}^{\infty}p(1-p)^j \int\limits_{0}^{\infty} \frac{1}{2}\log(hp(1-p)^jB_{max})e^{-h}dh,\\
\nonumber
=&\frac{1}{2}\log \left(1 + \sqrt{2p}B_{max} \right) - \\
\nonumber
&\sum\limits_{j=0}^{\infty}p(1-p)^j \int\limits_{0}^{\infty} \frac{1}{2} [ \log(h) + \log(p) + j\log(1-p)\\
\nonumber
&~~~~~~~~~~~~~~~~ + \log (B_{max)} ]e^{-h}dh,\\
\nonumber
\stackrel{(b)}{=}&0.25 + \frac{1}{4} \log (p) + \frac{1}{2}\log \left( B_{max} \right) \\
\nonumber
&+ \log \left( 1 + \frac{1}{\sqrt{2p}B_{max}} \right) + 0.29 - \frac{1}{2} \log (p)\\
\nonumber
& - \frac{1-p}{2p} \log (1-p) - \frac{1}{2} \log (B_{max}),\\
\nonumber
\stackrel{(c)}{\leq}&0.54 - \frac{1}{4} \log (p) + \frac{1}{2\ln 2}\frac{1}{\sqrt{2p}B_{max}}\\
\nonumber
& + \frac{1-p}{2p} \log \left(\frac{1}{1-p} \right),\\
\label{eqn:bound2}
\stackrel{(d)}{\leq}& \frac{1}{2}\log \left( 1 + \sqrt{2p}~k \right).
\end{align} 
where $(a)$ follows from the fact that removing 1 from the second log term results in an upper bound; $(b)$ follows because $\sum\limits_{j=0}^{\infty}p(1-p)^j=1$ and $\sum\limits_{j=0}^{\infty}jp(1-p)^j=\frac{1-p}{p}$. Also $\int\limits_{0}^{\infty}e^{-h}dh = 1$ and $\int\limits_{0}^{\infty} \log (h)e^{-h} = -0.29$, $(c)$ uses the identity $\ln (1+x) \leq x$, and finally $(d)$ follows from \eqref{eqn:recursion}.

\end{proof}

\section{}
\begin{proof}
The proposed strategy views any i.i.d. energy arrival process as Bernoulli with packet size $\delta$ and $p=0.5$, and uses the CFP. By Lemma 4, we have the following,
\begin{equation}
\label{second last}
T_{lb} \geq \frac{1}{2} \log \left( 1 + \delta \right) - 1.41
\end{equation}
We next bound the difference between $T_{ub}$ and the first term on the RHS of \eqref{second last} as follows, $T_{ub} - \frac{1}{2} \log \left( 1 + \delta \right)$
\begin{align}
\nonumber
 \stackrel{(a)}{=} &\frac{1}{2} \log \left( 1 + \sqrt{2} \sqrt{\mathbb{E}[X^2]} \right) - \frac{1}{2} \log \left( 1 + \delta \right),\\
\nonumber
\stackrel{}{=} & \frac{1}{2} \log \left( \frac{1 + \sqrt{2} \sqrt{\mathbb{E}[X^2]}}{1 + \delta} \right),\\
\nonumber
\stackrel{(b)}{\leq} & \frac{1}{2} \log \left( \frac{\sqrt{2} \sqrt{\mathbb{E}[X^2]}}{F_X^{-1}(0.5)} \right),\\
\label{eq:temp}
=  &  0.25 + \frac{1}{4} \log \left( \frac{\mathbb{E}[X^2]}{\left(F_X^{-1}(0.5) \right)^2} \right).
\end{align} 
where $(a)$ follows from Theorem 1, and $(b)$ follows because the numerator is greater than the denominator.

Rearranging the terms in \eqref{eq:temp}, we get, 
\begin{align}
\label{last}
\frac{1}{2} \log \left( 1 + \delta \right) \geq & T_{ub} - 0.25- \frac{1}{4} \log \left( \frac{\mathbb{E}[X^2]}{\left(F_X^{-1}(0.5) \right)^2} \right).
\end{align}

From \eqref{second last} and \eqref{last}, we have,
\begin{align}
&T_{lb} \geq T_{ub} - 1.67 - \frac{1}{4} \log \left( \frac{\mathbb{E}[X^2]}{\left(F_X^{-1}(0.5) \right)^2} \right).
\end{align}

\end{proof}

\bibliographystyle{IEEEtran}
\bibliography{references}
\end{document}

%% file: input.tex
\usepackage{amsfonts}
\usepackage{times}
\usepackage{latexsym}
\usepackage{amssymb}
\usepackage{amsmath}
\usepackage{cite}
\usepackage{verbatim}
\def\bb0{{\mathbb{0}}}

\def\bb{{\mathbf{b}}}

\def\bh{{\mathbf{h}}}

\def\b0{{\mathbf{0}}}

\def\b1{{\mathbf{1}}}

\def\bbE{{\mathbb{E}}}

\def\sf0{{\mathsf{0}}}

\def\SIR{{\mathsf{SIR}}}

\def\nn{\nonumber}